\documentclass[12pt]{article}

\usepackage{graphicx}
\graphicspath{{figures/}}
\DeclareGraphicsExtensions{.pdf,.png,.jpg}
\usepackage{verbatim}
\usepackage{caption}

\usepackage{lineno}

\usepackage{amsmath,amssymb,amsthm}
\usepackage{color}

\usepackage{fullpage}

\usepackage{hyperref}
\usepackage{cite}
\usepackage{url}
\newtheorem{theorem}{Theorem}
\newtheorem{lemma}{Lemma}

\newtheorem{property}{Property}

\theoremstyle{definition}
\newtheorem{definition}{Definition}

\newcommand{\remove}[1]{{}}
\newcommand{\note}[1]{{\color{red} #1}}
\newcommand{\bluenote}[1]{{\color{blue} #1}}

\newcommand{\T}{\mathcal{T}}

\newcommand{\I}{\mathcal{I}}
\newcommand{\ddd}{,\ldots ,}

\title{Reconfiguring Ordered Bases of a Matroid}
\author{%
Anna Lubiw$^*$
\and
Vinayak Pathak$^*$
}

\date{\today}

\begin{document}
%\linenumbers

\maketitle

\begin{abstract}
For a matroid with an ordered (or ``labelled'') basis, a basis exchange step removes one element with label $l$ and replaces it by a new element that results in a new basis, and with the new element assigned label $l$.
We prove that one labelled basis can be reconfigured to another if and only if for every label, the initial and final elements with that label lie in the same connected component of the matroid.  Furthermore, we prove that when the reconfiguration is possible, the number of basis exchange steps required is $O(r^{1.5})$ for a rank $r$ matroid.  For a graphic matroid we improve the bound to $O(r \log r)$.   
\end{abstract}

%%%%%%%%%%%%%%%%%%%%%%%%%%%%%%%%%%%%%%%%%%%%%%%%%%%%%%%%
%  INTRODUCTION
%%%%%%%%%%%%%%%%%%%%%%%%%%%%%%%%%%%%%%%%%%%%%%%%%%%%%%%%
\section{Introduction}
\label{sec:introduction}

Broadly speaking, ``reconfiguration'' is about changing one 
structure to a related one 
%solution of a problem to another 
via discrete steps.  
Examples include changing one triangulation of a point set in the plane to another via ``flips'', or changing one independent set of a graph to another by adding/deleting single vertices.
More broadly, other examples are to   
sort a list by swapping pairs of adjacent elements, or to change a basic solution of a linear program to an optimum one via pivots.
Central questions are: to test whether one configuration can be changed to another using the discrete steps; and to compute, or bound, the number of steps required.  
Reconfiguration problems have been considered for a long time, and explicit attention was drawn to complexity issues by Ito et al.~\cite{Ito-2011} and by van den Heuvel~\cite{van-den-Heuvel}. 

%These problems can be formulated as connectivity and shortest path problems in a \emph{reconfiguration graph} whose vertices correspond to configurations and where two vertices are joined by an edge when the one configuration can be changed into the other by a single elementary step.

A fundamental result about matroids is that any basis $B$ can be reconfigured to any other basis $B'$ using a sequence of \emph{basis exchange} steps.  In each \emph{basis exchange} one element of the current basis is replaced by a new element to yield a new basis.  The number of basis exchange steps needed to reconfigure $B$ to $B'$ is the difference $|B - B'|$.
 
In this paper we explore a ``labelled'' version of basis reconfiguration where the elements of each of the two initial bases are labelled with unique labels from $\{1, \ldots, r\}$, where $r$ is the rank of the matroid.  
For a basis exchange in this labelled setting, the new element gets the same label as the replaced element.
More precisely, if $l$ is the labelling function on $B$ and $e\in B$ is replaced with $e'\in E\setminus B$, then $e'$ gets assigned the label $l(e)$. 
In standard matroid terminology, a ``labelled'' basis is an ``ordered'' basis. 
% and we will use the two terms interchangeably.

We consider two questions:
\begin{enumerate}
\item When can one labelled basis of a matroid be reconfigured to another labelled basis of the matroid using a sequence of basis exchanges? 
%What properties guarantee that we can reconfigure one ordered basis to another?
\item What is the number of basis exchanges needed in the worst case?
\end{enumerate}

Note that such labelled reconfiguration is not always possible---for example, a matroid may have a single basis $B$, in which case there is no way to reconfigure from one labelling of $B$ to a different labelling.   

\medskip\noindent{\bf Results.} In this paper we prove that one labelled basis can be reconfigured to another if and only if for every label, the element with that label in the first basis and the element with that label in the second basis lie in the same connected component of the matroid.  
Equivalently, for a given basis, the permutations of labels achievable by sequences of basis exchange steps is a product of symmetric groups, one for each connected component of the matroid.  
%\note{Is this result known in the matroid literature?}

Furthermore, we prove that if one labelled basis can be reconfigured to another then 
$O(r^{1.5})$ basis exchange steps always suffice, where $r$ is the rank of the matroid.

In the special case of graphic matroids, 
%the general matroid terminology can be replaced by the terminology of graphs and spanning trees.
%In those terms, 
our problem is the following: given two spanning trees of an $n$-vertex graph, with the edges of the spanning trees labelled with the labels $\{1, 2, \ldots, n-1\}$, reconfigure the first labelled spanning tree to the second via basis exchange steps.
Reconfiguration is possible if and only if for every label, the edge with that label in the first spanning tree and the edge with that label in the second spanning tree lie in the same 2-connected component of the graph.
This was proved by Hernando et al.~\cite{HHMR03}.
In this case we can prove a tighter bound: $O(n \log n)$ basis exchange steps are always sufficient. 

The rest of the paper is organized as follows.  Subsections~\ref{sec:preliminaries} and~\ref{sec:background} contain notation and background.  In Section~\ref{sec:graphic-matroid} we prove the results for graphic matroids, and in Section~\ref{sec:general-matroids} we prove the results for general matroids.

%%%%%%%%%%%%%%%%%%%%%%%%%%%%%%%%%%%%%%%%%
\subsection{Preliminaries}
\label{sec:preliminaries}

For a matroid $M$ of rank $r$, a \emph{labelled basis} (or ``ordered basis'') is a tuple $\T = (B,l)$ where $B$ is a basis and $l:B\rightarrow \{1\ddd r\}$ is a function that assigns a unique label to each element of $B$.
For label $i$, the element with that label is given by $l^{-1}(i)$.
If a basis exchange step replaces $e\in B$ with $e'\in E\setminus B$, then $e'$ is assigned the label $l(e)$. 
Two bases are said to be the same if they have the same elements and the elements are assigned the same labels.
%A sequence of basis exchange steps is called an \emph{exchange sequence}.

%In the special case of a graphic matroid . . say what an exchange is.

\medskip\noindent{\bf Matroid Properties.}
For basic definitions and properties of matroids, see Oxley~\cite{Oxley-2011}.  In the remainder of this section, we summarize the properties that we will use.  Note that we are now referring to standard unlabelled matroids.
%\note{Some of these are in the .tex file, commented out.}

\remove{
\bluenote{Rather than ``flip'' say ``exchange''.  Also ``basis exchange graph'' instead of ``flip graph'', ``matroid base polytope''.  Be explicit about what an ``exchange'' is for spanning trees.  Define ``reconfiguration''.  Do not use ``flip distance''.} 

\note{Some notation.}
An ordered basis of a matroid is defined in the same way as an edge-labelled triangulation, i.e., an ordered basis $\T = (B, l)$ is a tuple where $B$ is a basis and $l:B\rightarrow \{1\ddd |B|\}$ is a function that assigns a unique label to each element of $B$. A reconfiguration step is now a basis exchange operation where the new element gets the same label as the replaced element, i.e., if $e\in B$ is replaced with $e'\in E\setminus B$, then $e'$ gets assigned the label $l(e)$. Two bases are said to be the same if they have the same elements and the elements are assigned the same labels. From now on, we will use \emph{flip} to denote one reconfiguration step and \emph{flip graph} to denote the reconfiguration graph. Since the labels are unique, we will frequently use the labels to refer to the elements they are assigned to. If a flip sequence $F$ transforms the ordered basis $\T_1$ to $\T_2$ such that label $i$ is assigned to $e_1$ in $\T_1$ and to $e_2$ in $\T_2$, we say that $F$ \emph{moves} label $i$ from $e_1$ to $e_2$. For label $i$, the element with that label is given by $l^{-1}(i)$. When we are talking specifically about graphic matroids, we will use the terms spanning tree and edge instead of basis and element respectively.
}

\remove{
The material in this section and in Section~\ref{section-more-matroids} is based on the exposition in the textbook by Oxley~\cite{Oxley-2011}.

Matroids are set systems that were originally defined to study an abstract theory of dependence. Given an $m\times n$ matrix, define the set system $M = (E, \I)$ such that $E = [1..n]$ and $I\in\I$ if and only if the set of columns corresponding to $I$ forms an independent set of vectors. This set system is an example of a matroid.

\begin{definition}
A matroid is a tuple $M = (E, \I)$ where $E$ is a set and $\I$ is a set of its subsets such that:
\begin{enumerate}
\item[(I1)] $\emptyset \in \I$.
\item[(I2)] If $I\in\I$ and $I'\subseteq I$ then $I'\in\I$.
\item[(I3)] If $I_1$ and $I_2$ are in $\I$ and $|I_2|>|I_1|$ then there exists an element $e\in I_2 - I_1$ such that $I_1\cup \{e\}\in\I$.
\end{enumerate}
\end{definition}

It is easy to see that these three properties are satisfied for the set system formed by independent sets of column vectors of a matrix. The sets of $\I$ are called \emph{independent sets} and property (I3) above is often called the \emph{independence augmentation} property. Many naturally-occurring set systems, often in unrelated areas, happen to be matroids. Since the linearly independent subsets of columns of a matrix form the independent sets of a matroid, it is apparent that sets of affinely independent points should also form a matroid. For example, given a set of points in a plane, the set of subsets of size at most two and subsets of three non-collinear points forms a matroid. A popular matroid that appears in graph theory is the one formed by the set of forests of a graph. Given a graph $G$, define the matroid $M = (E, \I)$ where $E$ conisists of all edges of $G$ and $I\in\I$ if and only if the edges corresponding to $I$ form a forest (i.e., do not contain a cycle) in $G$. A matroid formed this way is called a \emph{graphic matroid}.

A \emph{maximal independent set} of a matroid $M = (E, \I)$ is a set $I\in\I$ such that no superset of $I$ is a member of $\I$. A maximal independent set is also called a \emph{basis} (note the correspondence with bases of a vector space). The set of all bases is often denoted by the letter $\mathcal{B}$. The following two properties of bases are easy to show.

\begin{property}
\label{prop:rank}
If $B_1$ and $B_2$ are both bases of a matroid $M = (E, \I)$ then $|B_1| = |B_2|$.
\end{property}
\begin{proof}
For contradiction, assume $|B_1| < |B_2|$. Then from property (I3), there exists an element $e\in B_2 - B_1$ such that $B_1\cup\{e\}\in\I$. Thus $B_1$ was not maximal.
\end{proof}
}

Recall the basis exchange property of matroids:

\begin{property}
\label{prop:basis-exchange}
For any two bases $B_1$ and $B_2$ of a matroid $M$ %$M=(E, \I)$ 
and an element $x\in B_1 - B_2$, there exists $y\in B_2-B_1$ such that $(B_1 - \{x\})\cup\{y\}$ is also a basis.
\end{property}
%\begin{proof}
%Using property (I3) on $B_1\setminus\{x\}$ and $B_2$, we directly get the desired $y\in B_2-B_1$. 
%\end{proof}

By Property \ref{prop:basis-exchange} any basis of a matroid can be reconfigured into any other using at most $r$ basis exchange steps.

\remove{
Property~\ref{prop:rank} helps one define a useful quantity---the \emph{rank} of a matroid, denoted $r(M)$---as $r(M) = |B|$ for any basis $B$. Property~\ref{prop:basis-exchange}, which we will call the \emph{basis exchange property} from now on, already contains hints of reconfiguration. Indeed, given a matroid $M = (E, \I)$ and a basis $B$, define a reconfiguration step as replacing an element $e\in B$ with an element $e'\in E\setminus B$ such that the new set is also a basis. This defines a reconfiguration graph that has a vertex for each basis and edge for each reconfiguration step. The following is an easy consequence of Property~\ref{prop:basis-exchange}.

\begin{theorem}
The reconfiguration graph of bases of a matroid $M$ is connected and has diameter at most $r(M)$.
\end{theorem}
\begin{proof}
Given any two bases $B_1$ and $B_2$ of $M$, one application of the basis exchange property increases $|B_1\cap B_2|$ by one. Since the minimum and maximum possible values of $|B_1\cap B_2|$ are 0 and $r(M)$ respectively, the total number of applications of Property~\ref{prop:basis-exchange} required to reconfigure $B_1$ into $B_2$ in the worst case is $r(M)$.
\end{proof}

Matroid reconfiguration was first considered by Ito et al.~\cite{IDHPSUU11} and the theorem above is from that paper, although it is a well known result in the literature on matroids.

A basis of a matroid has properties similar to a spanning tree of a graph. One can also define a \emph{circuit} of a matroid to have properties similar to a cycle of a graph. A circuit of a matroid is defined as a minimal non-independent set, i.e., a set $C\subseteq E$ which is not independent but every subset of $C$ is independent. The set of circuits of a matroid satisfies some nice properties, including the following \emph{circuit exchange} property.
}

\remove{
Matroids also satisfy the following \emph{circuit exchange} property. \note{Check if you still use this.}

\begin{property}
\label{prop:circuit-exchange}
Given any two circuits $C_1$ and $C_2$, an element $e\in C_1\cap C_2$, and an element $f\in C_2 - C_1$, there exists another circuit $C_3$ such that $f\in C_3 \subset C_1\cup C_2 - \{e\}$.
\end{property}
}

The pairs of elements that can participate in exchanges can be characterized exactly.
Let $B$ be a basis.  Adding an element $e\notin B$ to $B$ creates a unique circuit $C$ called the 
\emph{fundamental circuit} of $e$ with respect to $B$.
\begin{property}~\cite[Exercise 5, Section 1.2]{Oxley-2011}
\label{prop:basis-exchange-advanced}
Given a basis $B$, and an element $e\notin B$, let $C$ be the fundamental circuit of $e$ with respect to $B$. 
Then the set of elements $e'$ such that $(B - \{e'\})\cup\{e\}$ is a basis are precisely the elements $e' \in C$.
% and removing any element of $C$ then gives us another basis $B'$.
\end{property}

The \emph{fundamental graph} of basis $B$, denoted $S_B$, is  the bipartite graph with a vertex for each element of the matroid, and an edge $(e,f)$ for every $e \in B$ and $f \in E - B$ such that $e,f$ form a basis exchange, i.e.~$(B - \{e\}) \cup \{f\}$ is a basis.  Note that the closed neighbourhood of $f \in E-B$ is $f$'s fundamental circuit.

\remove{
\begin{proof}
\note{The property was strengthened.  Check the proof.}
First of all, suppose, for contradiction, that adding $e$ to $B$ creates more than one circuit, $C$ and $C'$ being any two of them. Clearly, $e$ must be a part of both, and therefore, from Property~\ref{prop:circuit-exchange}, there exists another circuit $C''\subseteq C\cup C' -\{e\}$. This is not possible since $C\cup C'-\{e\}\subseteq B$ and therefore is independent.

Given that $C$ is unique, it is also clear that removing any element of $C$ will give us a basis. 
\end{proof}
}

\remove{
\smallskip\noindent{\bf Contraction}
\note{Rewrite the following.}
We would like to generalize the results from the previous section to more general matroids. However, we used many concepts about graphs without talking about their analogs in the world of matroids. The theory of matroids has very elegant generalizations for some graphic concepts and some graphic concepts have no matroidal analogs. In this section, we explore what can and cannot be generalized to matroids.

Our proof for graphic matroids relied on induction using edge contraction. Edge deletion and edge contraction are two commonly used operations for carrying out inductive arguments on graphs. Edge deletion has an obvious matroidal analog. Given a matroid $M = (E, \I)$, deleting a set $F\subseteq E$ of edges results in the matroid $M' = (E', \I')$, denoted by $M\setminus F$ such that $E' = E\setminus F$ and for each $I\in\I$, we add the set $I' = I\setminus F$ to $\I'$. Edge contractions also have a matroidal analog that can be elegantly defined using matroid duals.

Given a matroid $M = (E, \I)$, let $B^* = \{E - B~|~B\in B(M)\}$. Then it is known that $B^*$ is the set of bases of a matroid on $E$. The matroid whose set of bases is $B^*$ is called the \emph{dual} of $M$ and is denoted by $M^*$. One can then define edge contraction as follows. Given a matroid $M = (E, \I)$ and a set $F\subseteq E$ of edges, the matroid obtained on contracting the edges of $F$ is defined as $M/F = (M^*\setminus F)^*$, i.e., we delete the edges in the dual. This definition is justified because for a graphic matroid, this operation corresponds exactly to the operation of contracting the edges of the graph, i.e., if $M$ is a graphic matroid defined on a graph $G$, then the set of forests of $M/F$ is exactly the set of forests of the graph obtained on contracting $F$ in $G$. We will not prove this statement here, but a proof can be found, for example, in~\cite{Oxley-2011}.
}

For a matroid $M$ with element set $E$ and a set $T \subseteq E$, the matroid that results from contracting $T$ is denoted $M/T$.

\begin{property}
\label{prop:rank-after-contraction}
For matroid M, and sets $T \subseteq E$ and $X \subseteq E - T$
$$r_{M/T}(X) = r_M(X \cup T) - r_M (T).$$
\end{property}

This property implies that we can augment an independent set as follows:

\begin{property}
\label{prop:incremental-independent}
If $T$ is independent in matroid $M$ and $A$ is independent in $M/T$ then $T \cup A$ is independent in $M$.
\end{property}

%\smallskip\noindent{\bf Connectivity}
The concept of graph connectivity generalizes to matroids. For any matroid $M$, define the relation $\psi$ on the elements of $M$ by $e~\psi~f$ if and only if either $e = f$ or there exists a circuit of $M$ that contains both $e$ and $f$. It can be shown that the relation $\psi$ is an equivalence relation and we say that each equivalence class defines a \emph{connected component}.  In the case of a graphic matroid, the equivalence classes are the \emph{2-connected components} or \emph{blocks} of the graph.

We will use Menger's theorem for graphs, see~\cite[chapter 9]{Schrijver}, which says that the maximum number of vertex-disjoint paths from vertex $s$ to vertex $t$ is equal to the minimum number of vertices whose removal disconnects $s$ and $t$.
We will also use the generalization of Menger's theorem to matroids, 
which is known as Tutte's linking theorem~\cite{Tutte-linking}, see~\cite{Oxley-2011}. 
The statement of this theorem will be given where we use it in Section~\ref{section:matroid-tight-bound}. 

\remove{
We will need the following notation~\cite{Oxley-2011}.
For disjoint sets $X, Y \subseteq E$, the analogue of 
\begin{equation*}
\sqcap_M(X,Y) = r_M(X) + r_M(Y) - r_M(X \cup Y)
\end{equation*}
and  
\begin{equation*}
\kappa_M(X,Y) = \min \{ r_M (S) + r_M(E-S) - r_M (M) : X \subseteq S \subseteq E-Y \}.
\end{equation*}
}

%%%%%%%%%%%%%%%%%%%%%%%%%%%
\subsection{Relationship to Triangulations}
\label{sec:background}

Our present study of reconfiguring labelled matroid bases was prompted by %our 
related results on reconfiguring labelled triangulations~\cite{BLPV13}.

The set of triangulations of a point set (or a simple polygon) in the plane has some matroid-like properties. In particular, given a point set $P$, let $E$ be the set of all line segments $d$ such that the endpoints of $d$ are in $P$ and no other point of $P$ lies on $d$, and let $\I$ be the set of all subsets of pairwise non-crossing segments of $E$. Then %it is clear that 
the set system $(E, \I)$ is an independence system, but fails the third matroid axiom:  if $I_1$ and $I_2$ are in $\I$ and $|I_2| > |I_1|$, there does not always exist an element $e$ of $I_2 - I_1$ such that $I_1 \cup \{e\} \in \I$.  (As a simple example consider five points $1, 2, 3, 4, 5$ in convex position where $I_1$ contains segments $(1,3)$ and $(1,4)$ and $I_2$ contains segment $(2,5)$.) 
%satisfies properties (I1) and (I2) but not property (I3). 
However, the maximal sets in $\I$ all have the same cardinality---this is because maximal sets of non-crossing segments correspond to triangulations and all triangulations contain the same number of edges.
%do satisfy the property that if $B_1$ and $B_2$ are both maximal, then $|B_1| = |B_2|$---this is because maximal sets of non-crossing segments correspond to triangulations and all triangulations contain the same number of edges.

The failure of the third matroid axiom for triangulations means that the greedy algorithm does not in general find a minimum weight triangulation of a point set, and in fact the problem is NP-hard even for Euclidean weights~\cite{MR08}.

On the other hand, triangulations share some of the reconfiguration properties of matroids---any triangulation can be reconfigured to any other triangulation via a sequence of elementary exchange steps called ``flips'' where one segment is removed and a new segment (the opposite chord of the resulting quadrilateral, if it is convex) is added~\cite{BH09}.
Furthermore, there is evidence that results analogous to the ones we prove here for reconfiguring labelled matroid
bases carry over to labelled triangulations~\cite{BLPV13}.
%It seems possible that similar proof techniques might be used in both cases, and it is .
One  intriguing possibility is a more general result that encompasses both the case of matroids and of triangulations.

%%%%%%%%%%%%%%%%%%%%%%%%%%%%%%%%%%%
\remove{
\section{Problem statement - move to Intro}
An ordered basis of a matroid is defined in the same way as an edge-labelled triangulation, i.e., an ordered basis $\T = (B, l)$ is a tuple where $B$ is a basis and $l:B\rightarrow \{1\ddd |B|\}$ is a function that assigns a unique label to each element of $B$. A reconfiguration step is now a basis exchange operation where the new element gets the same label as the replaced element, i.e., if $e\in B$ is replaced with $e'\in E\setminus B$, then $e'$ gets assigned the label $l(e)$. Two bases are said to be the same if they have the same elements and the elements are assigned the same labels. From now on, we will use \emph{flip} to denote one reconfiguration step and \emph{flip graph} to denote the reconfiguration graph. Since the labels are unique, we will frequently use the labels to refer to the elements they are assigned to. If a flip sequence $F$ transforms the ordered basis $\T_1$ to $\T_2$ such that label $i$ is assigned to $e_1$ in $\T_1$ and to $e_2$ in $\T_2$, we say that $F$ \emph{moves} label $i$ from $e_1$ to $e_2$. For label $i$, the element with that label is given by $l^{-1}(i)$. When we are talking specifically about graphic matroids, we will use the terms spanning tree and edge instead of basis and element respectively.

It is easy to show now that the flip graph is not always connected. For example, let $G$ be a path on $n$ vertices. Clearly, $G$ has only one spanning tree, which is $G$ itself. Now this spanning tree gives us several ordered bases depending on how we label its edges. However, none of the elements of this basis can be flipped and thus any ordered basis of $G$ cannot be reconfigured to any other ordered basis.

This observation makes the following questions interesting:
\begin{enumerate}
\item What properties guarantee that we can reconfigure one ordered basis to another?
\item What is the number of reconfiguration steps needed in the worst case?
\end{enumerate}
}

%%%%%%%%%%%%%%%%%%%%%%%
\section{Reconfiguring ordered bases of a graphic matroid}
\label{sec:graphic-matroid}

In this section we concentrate on graphic matroids.  We characterize when one labelled spanning tree of an $n$-vertex graph can be reconfigured to another labelled spanning tree of the graph.  Our first result provides a bound of $O(n^2)$ exchange steps for the reconfiguration.  In the second subsection we improve this to $O(n \log n)$.
The first result (with the $O(n^2)$ bound) extends immediately to general matroids, but we give the details for the graphic matroid case for the sake of readers who may wish to learn only about labelled reconfiguration of spanning trees.

\subsection{When are two ordered bases reconfigurable?}
%In this section, we provide a polynomial time algorithm that takes a graph and two labelled spanning trees as its input and decides whether the first spanning tree can be reconfigured to the second via a sequence of basis exchanges.
%there exists a reconfiguration sequence
%\note{exchange sequence} that transforms one to the other.

\begin{theorem}
\label{theorem-spanning-tree-characterization}
Given two labelled spanning trees $\T_1$ and $\T_2$ of an $n$-vertex graph $G$, we can reconfigure one to the other if and only if for each label $i$, the edge with that label in $\T_1$ and the edge with that label in $\T_2$ lie in the same 2-connected component of $G$. Moreover, when reconfiguration is possible, it can be accomplished with $O(n^2)$ basis exchange steps. 
\end{theorem}
\begin{proof}
The `only if' direction is clear because an exchange of edge $e$ to $e'$ can be performed only if both $e$ and $e'$ lie in the same 2-connected component. For the `if' direction, we will provide an explicit exchange sequence to reconfigure $\T_1$ to $\T_2$.

First, pick any (unlabelled) spanning tree $B$ and reconfigure both $\T_1$ and $\T_2$ to $B$ while ignoring the labels. This takes $O(n)$ exchange steps. Let $\sigma_1$ be the sequence of exchanges that reconfigures $\T_1$ to $B$ and $\sigma_2$ be the sequence that reconfigures $\T_2$ to $B$. Obviously, the labels of the edges of $B$ obtained from the two sequences will not match in general. %We show, below, 
Below, we give an exchange sequence $\sigma$ of length at most $O(n^2)$ to rearrange the labels in $B$. Thus performing $\sigma_1$ followed by $\sigma$ followed by the reverse of $\sigma_2$ reconfigures $\T_1$ to $\T_2$ with $O(n^2)$ exchange steps.

The problem is now reduced to the following: given one spanning tree $B$ and two labellings of it,  $\T_1 = (B, l_1)$ and $\T_2 = (B, l_2)$, reconfigure $\T_1$ to $\T_2$ using $O(n^2)$ exchanges. We will do this by repeatedly swapping labels. 
More precisely, let $i$ be a label, let $e_1 = l_1^{-1}(i)$ be the edge with label $i$ in $\T_1$ and let $e_2 =l_2^{-1}(i)$ be the edge with label $i$ in $\T_2$, and suppose that $e_1 \ne e_2$.  We will show that with $O(n)$ exchanges we can swap the labels of $e_1$ and $e_2$ in $\T_1$ while leaving all other labels unchanged.  This moves label $i$ to the correct place for $\T_2$, and repeating over all labels solves the problem and takes $O(n^2)$ exchanges. 

Since $e_1$ and $e_2$ lie in the same 2-connected component of $G$ (by hypothesis), there must exist a cycle $C$ of $G$ that goes through both $e_1$ and $e_2$.  
Let $t$ be the number of edges of $C \setminus B$.  Then $t \ge 1$.  We will argue by induction on $t$ that $6t-3$ exchanges suffice to swap the labels of $e_1$ and $e_2$.
%We now argue by induction on the number of edges of $C \setminus B$. 
Let $f$ be an edge of $C \setminus B$.

If $f$ is the only edge of $C \setminus B$, then we can perform the swap directly: use Property~\ref{prop:basis-exchange-advanced} to exchange $e_1$ with $f$, $e_2$ with $e_1$, and, finally, $f$ with $e_2$. This sequence returns us to $B$ and swaps the labels of $e_1$ and $e_2$ while leaving all other labels unchanged.  Thus the case of $t=1$ can be solved with 3 exchanges.

More generally,  $f$ is not the only edge of $C \setminus B$.  Adding $f$ to $B$ creates a cycle $C'$ that must contain an edge $f' \in B \setminus C$.  Using Property~\ref{prop:basis-exchange-advanced} we exchange $f$ with $f'$.  The result is a new spanning tree $B'$ whose intersection with $C$ is strictly increased, so $C \setminus B'$ is smaller. 
Also note that $B'$ contains $e_1$ and $e_2$. 
By induction, we can swap the labels of $e_1$ and $e_2$ in $B'$ with at most $6(t-1) -3$ exchanges, leaving all other labels unchanged.  After that, we exchange $f'$ and $f$ to return to $B$ with original labels except that the labels of $e_1$ and $e_2$ are now swapped.  The total number of exchanges is at most $6(t-1) -3 + 6 = 6t -3$
\end{proof}

%%%%%%%%%%%%%%%%%%%%%%%%%%%%%%
\subsection{Tightening the bound}
\label{section:graphic-tight-bound}
In this section we show 
that the $O(n^2)$ bound on the number of exchanges from the previous section can be improved. Note that the common spanning tree $B$ we chose in the previous section, to reconfigure both $\T_1$ and $\T_2$ to, was completely arbitrary. We could have perhaps chosen a spanning tree that made the task of swapping labels easier. That is precisely what we do in this section.

Observe that it is sufficient to consider a 2-connected graph $G$ since for a general graph we can just repeat the argument inside each of its 2-connected components. 
We will construct a spanning tree $B$ of $G$ whose fundamental graph (as defined in Section~\ref{sec:preliminaries}) has diameter $O(\log n)$.  This proves our result based on:

\begin{lemma}
Let $G$ be a 2-connected graph, and $B$ be a spanning tree of $G$
whose fundamental graph $S_B$ has diameter $d$.  Then for any two edges of $B$ there is an exchange sequence of length $O(d)$ that swaps the labels of those two edges while leaving other labels unchanged.
%
%we can swap the labels of any two edges of $B$ with an exchange sequence of length at most $O(d)$.
\label{lemma:swap}
\end{lemma}
\begin{proof}
Let $e$ and $f$ be two edges of $B$, and suppose the shortest path between them in $S_B$ has length $t<d$.
We will prove by induction on $t$ 
that we can swap the labels of $e$ and $f$ with $3t-3$ exchanges. 
Let $e_1$ and $e_2$ be the two vertices that occur immediately after $e$ on a shortest path from $e$ to $f$ in $S_B$. 
Then the cycle formed by adding $e_1$ to $B$ contains both $e$ and $e_2$ by Property~\ref{prop:basis-exchange-advanced}, and by the same property, we can perform the following exchanges:
$e_1$ with $e$, then $e_2$ with $e_1$, then $e$ with $e_2$.  
This exchange sequence returns us to $B$ while swapping the labels of $e$ and $e_2$, and leaving other labels unchanged. 
In the basis case of the induction $t=2$ and $e_2=f$ and this completes the swap with $3$ exchanges.
In the general case, 
the distance between $e_2$ and $f$ in $S_B$ is $t-2$.  By induction, we can swap the labels of $e_2$ and $f$ with $3(t-2)-3$ exchanges.   After that we repeat the first 3 exchanges to complete the swap of the labels of $e$ and $f$.  
All other labels are unchanged.
The total number of exchanges is $6 + 3(t-2) -3 = 3t -3$. 
\end{proof}

Thus it remains to construct a spanning tree $B$ such that the diameter of $S_B$ is $O(\log n)$. 
We will construct our spanning tree by repeatedly contracting cycles.
For a 2-connected graph $G$ with edge set $E(G)$ and for a cycle $C \subseteq E(G)$, 
let $G/C$ denote the graph obtained by contracting $C$.
Note that $E(G/C) = E(G) - E(C)$. Contracting $C$ 
%divides $G$ into 2-connected components, or
creates  \emph{blocks} that are the maximal subgraphs of $G/C$ that are 2-connected. 
%Note that any two edges $e, e'\in E(G/C)$ are in different blocks if and only if all paths in $G$ containing both $e$ and $e'$ pass through at least one vertex of $C$. %We first prove a lemma that helps us formulate an inductive argument.
In order to get a bound on the diameter of $S_B$ we will need the following:

\begin{lemma}
\label{lemma-graphic-matroid-induction}
Any 2-connected graph $G$ with $m$ edges contains a cycle $C$ such that all blocks of $G/C$ have at most $m/2$ edges.
\end{lemma}
\begin{proof}
Let $C$ be the cycle that minimizes the size of the largest block obtained upon contracting it. 
We claim that all blocks of $G/C$ have at most $m/2$ edges.  Suppose not.  Then there exists a block $H$ of size bigger than $m/2$.  We will derive a contradiction.
%If all blocks in $G/C$ are of size at most $m/2$, then we are done. Otherwise there exists a block $H$ of size bigger than $m/2$.
 
Some edges of $E(H)$ are incident on vertices of $C$ in $G$; let those vertices be $\{v_1\ddd v_k\}$ in clockwise order along $C$. There are two paths between $v_1$ and $v_2$ along the cycle $C$ in $G$, one clockwise and one counterclockwise. Let $P$ be the one that is counterclockwise and thus contains all vertices of $\{v_1\ddd v_k\}$. There also exists a path $P'$ between $v_1$ and $v_2$ that uses only the vertices of $H$. We define $C'$ to be $P\cup P'$. We claim that the size of the largest block of $G/C'$ is smaller than the size of the largest block of $G/C$, hence reaching a contradiction.

First, note that no block of $G/C'$ contains an edge of $H$ and an edge not in $H$.  For consider edges $e$ and $e'$ in $G/C'$ with $e$ in $H$ and $e'$ not in $H$.  Any path from $e$ to $e'$ in $G$ must go through a vertex of $C$, and in particular, must go through a vertex of $P$, since $P$ contains all vertices of $C$ that have an edge of $H$ incident to them.  Because $C'$ contains all of $P$, therefore $e$ and $e'$ are in different blocks of $G/C'$.   

\remove{
First, note that no block of $G/C'$ contains an edge of $H$ and an edge of a block of $G/C$ different from $H$.  
%For suppose that block $H'$ of $G/C$ contains an edge $e'$ that is in the same block of $G/C'$ as edge $e$ of $H$.  Then there is a path that contains $e$ and $e'$ but does not use any vertex of $P$.  
In other words, if  $H'$ is a block of $G/C$ different from $H$, then for all $e\in H$ and $e'\in H'$, if $e$ and $e'$ lie in $G/C'$, then they must be in different blocks of $G/C'$. 
To justify this, note that any path from $e$ to $e'$ in $G$ 
must go through a vertex of $C$, and in particular, must go through a vertex of $P$, since $P$ contains all vertices of $C$ that have an edge of $H$ incident to them.  Because $C'$ contains all of $P$, therefore $e$ and $e'$ are in different blocks of $G/C'$.   
%
%Suppose, for contradiction, that they lie in the same component. Then there exists a path that contains both $e$ and $e'$ but none of the vertices from $P$. But $P$ contains all vertices of $C$ that have an edge of $H$ incident on them. Thus the path containing both $e$ and $e'$ cannot contain any vertex of $C$, which means $e$ and $e'$ were in the same block of $G/C$.
}

Now $G/C'$ contains two kinds of blocks: those that contain edges of $H$ and those that do not. Blocks of the first kind must have size at most the size of $H$ from the argument above. In fact, they must be strictly smaller than $H$ since $C'$ contains at least one edge of $H$. Blocks of the second kind must also be smaller than $H$ since at worst, such a block contains all edges of $G$ that are not edges of $H$, and there are at most $m/2$ such edges.
%which has at most $m/2$ edges.
\end{proof}

We are now ready to construct our spanning tree $B$.
\begin{lemma}
\label{lemma-graph-good-tree}
Given a 2-connected graph $G$, there exists a spanning tree $B$ such that the diameter of $S_B$ is $O(\log n)$.
\end{lemma}
\begin{proof}
%We start out by setting $B$ to be the empty set and gradually add edges to it. 
The algorithm for constructing $B$ proceeds in iterations $i=1, 2, \ldots$.  In iteration $i$ we will add a set $B^i$ to $B$.   In the first iteration we find the cycle $C$ of Lemma~\ref{lemma-graphic-matroid-induction} such that all blocks of $G/C$ are of size at most $m/2$. 
Let  $B^1$ be all edges but one of $C$, and contract 
%Next, we add all edges but one of $C$ to $B$ and contract 
those edges (equivalently, contract $C$)  
thus breaking the graph $G$ into several blocks. In general, in any iteration, we start with the collection of blocks produced in the previous iteration,
%the previous iteration resulted in, 
contract a connected set of edges inside each block, and add %some of the contracted 
those edges to $B^i$. After each contraction we eliminate loops and parallel edges.  Each iteration reduces the number of vertices of $G$ by contracting a set of edges. The algorithm terminates once the graph $G$ is left with just one vertex.

We now describe what happens to one of the blocks that is dealt with in iteration $i$.  
For $i \ge 2$, let $H^{i-1}$ be a block at the beginning of iteration ${i-1}$, and let $B^{i-1}$ be the edges of $H^{i-1}$ that we contract and add to $B$ in the $(i-1)^{\rm st}$ iteration.  Let $b$ be the vertex that $B^{i-1}$ gets contracted to, and let $H^{i}$ be one of the blocks formed by the contraction. 
Then in iteration $i$, 
we pick %$B^{i}$, 
the edges of $H^{i}$ to be contracted and added to $B^i$, as follows. 
Let $C^{i}$ be the cycle of Lemma~\ref{lemma-graphic-matroid-induction} for $H^{i}$.  Observe that $C^i$ may or may not include vertex $b$.  Let $e$ be an edge of $C^i$.  As in the first iteration, we will add to $B^i$ all edges of $C^i$ except $e$. However, we may need to add more edges in order to maintain connectivity of $B$, and to ensure that $S_B$ has small diameter. 

In $H^{i-1}$ the sets of edges $B^{i-1}$ and $C^i$ are disjoint.  
We will use the fact that $H^{i-1}$ is 2-connected and apply Menger's Theorem, see~\cite[chapter 9]{Schrijver}, to find a set of edges that connect $B^{i-1}$ and $C^i$.
More precisely, let $s$ be a new vertex adjacent to all endpoints of edges in $B^{i-1}$ and $t$ be a new vertex adjacent to all endpoints of edge in $C^i$. 
Because $H^{i-1}$ is 2-connected, we must remove at least 2 vertices to disconnect $s$ and $t$. 
Then, by Menger's Theorem, there are two internally vertex-disjoint paths from $s$ to $t$.  
Let $P$ be a minimal set of edges of $H^{i-1}$  that form two such paths.  Note that if $C^i$ includes vertex $b$, one or both of the paths will have no edges of $H^{i-1}$.  See Figure~\ref{fig:graph-basis}.

\begin{figure}
\centering
\includegraphics[width=6in]{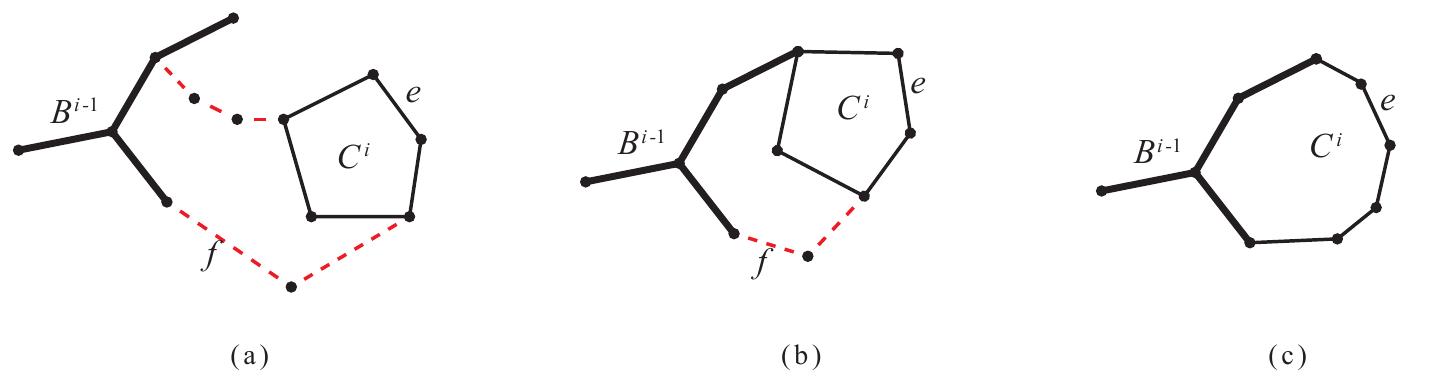}
\caption{Adding a minimal set of edges $P$ (in dashed red) to create a 2-connection between $B^{i-1}$ (thick edges) and $C^i$ (thin edges). $P$ may consist of two, one, or zero disjoint paths as shown in (a), (b), or (c), respectively.}
\label{fig:graph-basis}
\end{figure}

Observe that the two paths of $P$ go from two distinct vertices that are joined by a path in $B^{i-1}$ to two distinct vertices that are joined by a path in $C^i - \{e\}$.  Thus, a cycle, $D$, is formed by $P$ together with a non-empty subset of $B^{i-1}$ and a non-empty subset of $C^i - \{e\}$.  
Let $f$ be an edge of $P$ (if $P$ is non-empty).  By minimality of $P$, there is no cycle in $P - \{f\}$. Add to  $B^i$ the set $(C^i - \{e\}) \cup (P - \{f\})$.  This completes the description of how we handle one block in iteration $i$.   We handle other blocks the same way, adding further edges to $B^i$ for each other block.  
This completes the description of the algorithm.

To establish the correctness of the algorithm, first note that after each iteration $B = \cup_{j=1}^i B^j$ is connected and contains no cycle.  %\note{This depends on minimality of $P$.}
Thus, when the algorithm terminates, $B$ spans all vertices of $G$ and contains no cycle.
It remains to prove that the diameter of $S_B$ is $O(\log n)$.   
In each iteration of the algorithm, we reduce the size of each block by at least half, and thus the 
%Since each iteration reduces the size of each block by at least half, it is clear that the 
algorithm terminates in at most $O(\log n)$ iterations. 
To complete the proof we will show that for each $i$, every edge in $B^i$ has a path of length $O(1)$ in $S_B$ to some edge in $B^{i-1}$.

Referring to the step of the construction described above, note that the fundamental circuit of $e$ in $B^i$ contains all of $C^i$.  Thus $e$ is joined by an edge of $S_B$ to every edge of $C^i$.  If $P$ is empty, then the fundamental circuit of $e$ also includes an edge of $B^{i-1}$ and we are done.  Otherwise,  the fundamental circuit of $f$ in $B$ contains all of $D$, which includes all of $P$ together with at least one edge of $C^i$ and at least one edge of $B^{i-1}$.  Thus $f$ is joined by an edge of $S_B$ to every edge of $P$, and to at least one edge of $C^i$ and to at least one edge of $B^{i-1}$.
Therefore, in $S_B$ every edge of $B^i$ is within distance 4 of some edge of $B^{i-1}$.
\end{proof}

%This gives 

Lemmas~\ref{lemma:swap}, \ref{lemma-graphic-matroid-induction}, and \ref{lemma-graph-good-tree} give us the following strengthened form of Theorem~\ref{theorem-spanning-tree-characterization}.
\begin{theorem}
If $\T_1$ and $\T_2$ are two labelled spanning trees of an $n$-vertex graph $G$ and for each label $i$, the edge with that label in $\T_1$ and the edge with that label in $\T_2$ lie in the same 2-connected component of $G$ then
we can reconfigure $\T_1$ to $\T_2$ using  
$O(n\log n)$ basis exchange steps.
\label{thm:graphic-bound}
\end{theorem}

%%%%%%%%%%%%%%%%%%%%%%%
\section{Reconfiguring ordered bases of a general matroid}
\label{sec:general-matroids}

In this section we turn to general matroids.  
We generalize the result of the previous section that 
characterizes when one labelled basis of a rank $r$ matroid can be reconfigured to another labelled basis of the matroid and prove a bound of $O(r^2)$ exchange steps for the reconfiguration.  
In the second subsection we improve this bound to $O(r^{1.5})$, a weaker bound than was possible for graphic matroids. 

%\subsection{Adding more tools to our matroidal toolbox}
\subsection{Connectivity}

Our goal is to follow the proof of Theorem~\ref{theorem-spanning-tree-characterization}, which used edge contraction, cycles, and 2-connectivity in a graph.
Observe that contraction of edges in a graph corresponds to contraction of elements in a matroid, cycles in a graph correspond to circuits in a matroid, and 2-connectivity in a graph corresponds to connectivity in a matroid (every pair of elements is contained in some circuit). 

With these correspondences,  
it is easy to check that every step of the proof of Theorem~\ref{theorem-spanning-tree-characterization} goes through for matroids and thus we get the following theorem, which we state without proof.

\begin{theorem}
\label{theorem-matroid-characterization}
Given two labelled bases $\T_1$ and $\T_2$ of a rank $r$ matroid $M$, we can reconfigure one to the other if and only if for each label $i$, the element with that label in $\T_1$ and the element with that label in $\T_2$ lie in the same block of $M$. 
%Moreover, the exchange sequence, if it exists, is of length at most $O(n^2)$.
Moreover, when reconfiguration is possible, it can be accomplished with $O(r^2)$ basis exchange steps. 
\end{theorem}

%%%%%%%%%%%%%%%%%%%%%%%%%%%%%%%%%%
\subsection{A tighter bound for general matroids}
\label{section:matroid-tight-bound}

In order to tighten the $O(r^2)$ bound on the number of exchanges needed for labelled basis reconfiguration in general matroids, we would like to follow the approach we used for graphic matroids in Section~\ref{section:graphic-tight-bound}.  
Lemma~\ref{lemma:swap} carries over directly so it suffices to build a basis $B$ whose fundamental graph $S_B$ has small diameter. 
For graphic matroids of rank $r$, we achieved diameter $O(\log r)$ but for general matroids we will only achieve a weaker bound of diameter $O(\sqrt r)$. 
Our starting point for graphic matroids was  
Lemma~\ref{lemma-graphic-matroid-induction} which proved that there is a cycle whose contraction cuts the size of a block in half.  For general matroids the analogous result is conjectured to be true, but we must rely on the following weaker result,  attributed to Seymour, and with an explicit proof in~\cite[Corollary 1.4]{MRWW05}.

%With these definitions, many of the proofs from the previous section can be copied almost verbatim to give analogous proofs for matroids except Lemma~\ref{lemma-graphic-matroid-induction}. However, we can prove a matroidal analog of Lemma~\ref{lemma-matroid-induction} with looser bounds, which finally leads to a worst-case upper bound of $O(n^{1.5})$ on the flip distance.
%In order to tighten the bound, we use the following lemma, attributed to Seymour, and with an explicit proof in~\cite[Corollary 1.4]{MRWW05}.

\begin{lemma}
\label{lemma-matroid-induction}
Let $C$ be the biggest circuit of a connected matroid $M$. Then the biggest circuit of $M/C$ is strictly smaller than $C$.
\end{lemma}

%This gives us the following theorem.
Using this lemma we can follow the approach we used for graphic matroids to prove the following bound.

\begin{theorem}
\label{theorem-matroid-bound}
If $\T_1$ and $\T_2$ are two labelled bases of a rank $r$ matroid  $M$ and for each label $i$, the element with that label in $\T_1$ and the element with that label in $\T_2$ lie in the same connected component of $M$ then
we can reconfigure $\T_1$ to $\T_2$ using  
$O(r^{1.5})$ basis exchange steps.
%
%Given two labelled bases $\T_1$ and $\T_2$ of a matroid $M$, we can reconfigure one to the other if and only if for each label $i$, the element with that label in $\T_1$ and the element with that label in $\T_2$ lie in the same connected component of $M$. Moreover, the exchange sequence, if it exists, is of length at most $O(r^{1.5})$ where $r$ is the rank of the matroid.
\end{theorem}
\begin{proof}
Following the idea of the proof of Theorem~\ref{thm:graphic-bound}, 
we will prove the bound by 
showing that every connected component of the matroid has 
a basis $B$ whose fundamental graph $S_B$ has diameter $O(\sqrt{r})$. We do this using almost exactly the  algorithm of Lemma~\ref{lemma-graph-good-tree} with one difference: instead of picking the cycle of Lemma~\ref{lemma-graphic-matroid-induction} in each block in each iteration, we will pick the biggest circuit and use Lemma~\ref{lemma-matroid-induction}.

As before, the algorithm for constructing $B$ proceeds in iterations $i=1, 2, \ldots$.  In iteration $i$ we will add a set $B^i$ to $B$.   In the first iteration we find the the largest circuit $C$ in $M$ and let $B^1$ be all elements but one of $C$.  We contract 
%Next, we add all edges but one of $C$ to $B$ and contract 
those elements,  
thus breaking the matroid into several connected components. In general, in any iteration $i$, we start with the collection of connected components produced in the previous iteration,
%the previous iteration resulted in, 
contract some elements inside each component, and add %some of the contracted 
those elements to $B^i$. After each contraction we eliminate loops and parallel elements.  
Each iteration reduces the number of elements of $M$ and the algorithm terminates when no elements are left.

We now describe what happens to one of the connected components that is dealt with in iteration $i$. 
For $i \ge 2$, let $H^{i-1}$ be a connected component at the beginning of iteration ${i-1}$, and let $B^{i-1}$ be the elements  of $H^{i-1}$ that we contract and add to $B$ in the $(i-1)^{\rm st}$ iteration.  
Let  $H^{i}$ be one of the connected components formed by the contraction. 
Then in iteration $i$, 
we pick %$B^{i}$, 
the elements of $H^{i}$ to be contracted and added to $B^i$, as follows. 

Let $C^i$ be the biggest circuit of $H^i$ and let $e$ be an element of $C^i$.
As before, we will put $C^i - \{e\}$ into $B^i$, but, as before, we may need to add elements to connect the independent set $B^{i-1}$ with the circuit $C^i$. 
We will be working in the matroid $H^{i-1}$ which we abbreviate as $H$.
%We first discuss the case where no extra elements are needed.  
 
Since $C^i$ is a circuit in $H^{i}$, we have $ |C^i| -1 = r_{H^i}(C^i) $.
Now, $H^i$ is a connected component of $H/ B^{i-1}$, so by Property~\ref{prop:rank-after-contraction},  
$r_{H^i}(C^i) = r_H(C^i \cup B^{i-1}) - r_H(B^{i-1})$.  
Thus $r_H(C^i \cup B^{i-1}) = |B^{i-1}|  + |C^i| - 1$, which means that $C^i \cup B^{i-1}$ has co-rank 1 and contains a unique circuit.

If $C^i$ is independent in $H$ then $C^i \cup B^{i-1}$ has a circuit formed by the elements of $C^i$ together with at least one element of $B^{i-1}$.  
%$|C^i| = r_H(C^i)$ so $r_H(C^i \cup B^{i-1}) = r_H(B^{i-1}) + r_H(C^i) - 1$.
In this case we add to $B^i$ the set $C^i - \{e\}$.  Observe that this set is independent in $H$
and that the fundamental circuit of $e$ contains all elements of $C^{i-1}$ and at least one element of $B^{i-1}$.   In the graphic case, this corresponds to Figure~\ref{fig:graph-basis} (right).
We will prove below that $B^i$ has the properties we need. 

Otherwise, $C^i$ is not independent in $H$.  In this case, the circuit in $C^i \cup B^{i-1}$ is $C^i$, and 
%Then  $r_H(C^i)=|C^i| -1$, and so $r_H(C^i \cup B^{i-1}) = r_H(B^{i-1}) + r_H(C^i)$.  
%In this case 
we will need to add more elements to $B^i$ in order to connect $B^{i-1}$ with $C^i$.
We will use the matroid analogue of Menger's Theorem which is known as Tutte's Linking Theorem~\cite{Tutte-linking}.
This theorem applies to two disjoint sets of elements in a matroid.  In our case the matroid is $H$, 
and the disjoint sets are $B^{i-1}$ and $C^i - \{e\}$, both of which are independent in $H$.
To ease notation, let $A = C^i - \{e\}$.  

The analogue of a separating set of vertices is $\kappa_H ( B^{i-1}, A)$, defined as the minimum, over sets $X$ that contain $B^{i-1}$ and exclude $A$, of $r_H(X) + r_H(E - X) - r_H(E)$.
%, where $E$ is the set of elements of $M$.  
Since $H$ is connected, this minimum is at least 1.  In notation, we have:

$$\kappa_H ( B^{i-1}, A) = \min_{B^{i-1} \subseteq X \subseteq E \setminus A}  r_H(X) + r_H(E - X) - r_H(E) \ge 1.$$

The analogue of vertex-disjoint paths in Menger's theorem is 
$\sqcap_{H/P} (B^{i-1}, A)$, defined as the maximum over sets $P$, of 
$r_{H/P} (B^{i-1}) + r_{H/P} (A) - r_{H/P} (B^{i-1} \cup A)$.

According to the version of Tutte's Linking Theorem stated as equation (8.16) in Oxley~\cite{Oxley-2011}, 
there exists a set of elements $P$ such that 

\begin{equation}
\label{eq:tutte}
\sqcap_{H/P} (B^{i-1}, A) = \kappa_{H} ( B^{i-1}, A). 
\end{equation}

\remove{
In this equation the right-hand side is the analogue of vertex separation in Menger's theorem.  
In particular, $\kappa_H ( B^{i-1}, A)$ is defined (see Section~\ref{sec:preliminaries})  as the minimum, over sets $X$ that contain $B^{i-1}$ and exclude $A$, of $r_H(X) + r_H(E - X) - r_H(E)$, where $E$ is the set of elements of $M$.  Since $H$ is connected, this minimum is at least 1.  In notation, we have:

$$\kappa_H ( B^{i-1}, A) = \min_{B^{i-1} \subseteq X \subseteq E \setminus A}  r_H(X) + r_H(E - X) - r_H(E) \ge 1$$

The left-hand side of equation~(\ref{eq:tutte}) is the analogue of vertex-disjoint paths in Menger's theorem.  
$\sqcap_{H/P} (B^{i-1}, A)$ is defined  (see Section~\ref{sec:preliminaries}) as 
$r_{H/P} (B^{i-1}) + r_{H/P} (A) - r_{H/P} (B^{i-1} \cup A)$.
}

As in the proof of the generalization of Tutte's Linking Theorem due to Geelen, Gerards, and Whittle~\cite[proof of Theorem 8.5.7]{Oxley-2011},
we will choose $P$ to be a minimal  set such that equation~(\ref{eq:tutte}) holds.
As shown in their proof, such a minimal $P$ is independent, and is \emph{skew} to $B^{i-1}$ and $A$.  Two sets are \emph{skew} if the rank of their union is the sum of their ranks.
In our case, $B^{i-1}$ and $A$ are  independent, so skewness implies that $B^{i-1} \cup P$ and $A \cup P$ are independent.
  
Applying Property~\ref{prop:rank-after-contraction} to $\sqcap_{H/P} (B^{i-1}, A)$ yields

\begin{equation}
\label{eq:P}
\sqcap_{H/P} (B^{i-1}, A) = r_H(B^{i-1} \cup P) + r_H(A \cup P) - r_H(P) - r_H(B^{i-1} \cup A \cup P).
\end{equation}
When $P$ is independent and skew to $B^{i-1}$ and $A$ this becomes:
\begin{equation}
\label{eq:P2}
\sqcap_{H/P} (B^{i-1}, A) = |B^{i-1}| + |A| + |P| - r_H(B^{i-1} \cup A \cup P).
\end{equation}  
From this, it is clear that if the value of equation~(\ref{eq:tutte}) is greater than 1, then we can delete elements of $P$ to obtain a minimal $P$ with $\sqcap_{H/P} (B^{i-1}, A) = 1$.  

For the remainder of the proof we 
define $P$ to be a minimal set with $\sqcap_{H/P} (B^{i-1}, A) = 1$.  
%Observe that $P \ne \phi$ since $B^{i-1} \cup A$ is independent.  
From equation (\ref{eq:P2}) we know that $B^{i-1} \cup A \cup P$ has co-rank 1.  
%In particular, this means that $P \ne \phi$ since $B^{i-1} \cup A$ is independent.  
There must be a unique circuit $D$ in $B^{i-1} \cup A \cup P$ and $D$ must contain all elements of $P$ (by minimality of $P$) and at least one element of $B^{i-1}$ (since $A \cup P$ is independent) and at least one element of $A$ (since $B^{i-1} \cup P$ is independent).

Let $f$ be an element of $P$ and add to $B^i$ the set $(C^i - \{e\}) \cup (P - \{f\})$.  Observe that this set is independent in $H$.  The fundamental circuit of $e$ contains all the elements of $C^i$, and the fundamental circuit of $f$ contains all the elements of $P$ and at least one element of $B^{i-1}$ and at least one element of $C^i$.

Before proceeding with the proof we will mention why the above analysis was separated into  two cases depending on whether or not $C^i$ is independent in $H$. 
%We seem to have made no use of the assumption that $C^i$ is a circuit in $H$.   But observe 
Observe that if $C^i$ is independent in $H$, then the minimal set $P$ that connects $B^{i-1}$ and $C^i - \{e\}$ is  in fact $P = \{e\}$ itself, and when we choose $f$ to be an element of $P$ then $f=e$.  It would be rather strange in this case (though strictly speaking, correct) to say that we add $(C^i - \{e\}) \cup (P - \{f\})$ to $B^i$.  That is why we treated the case where $C^i$ is independent in $H$ as a separate case.  

To complete the proof of the Theorem we must show that the final $B$, defined as $\cup B^i$,  is a basis and that the diameter of $S_B$ is $O(\sqrt r)$.  
Since we continue until everything is contracted away, it is clear that $B$ spans the matroid. 
Because each $B^i$ is independent after contracting all $B^j$ for $j < i$, Property~\ref{prop:incremental-independent} implies that $B$ is independent in the matroid $M$.  Thus $B$ is a basis.

We now analyze the diameter of $S_B$.
We first observe that the algorithm terminates in $O(\sqrt{r})$ iterations. This is because the number of possible iterations for which there exists a block containing a circuit of size $\Omega(\sqrt{r})$ can be at most $O(\sqrt{r})$ and once the size of the biggest circuit in each block has been reduced to $O(\sqrt{r})$, there can be at most $O(\sqrt{r})$ more iterations.

To complete the proof we will show that for each $i$, every edge in $B^i$ has a path of length $O(1)$ in $S_B$ to some edge in $B^{i-1}$.  
We will refer to the step of the construction described above.
As noted above, the fundamental circuit of $e$ in $B$ contains all of $C^i$.  Thus $e$ is joined by an edge of $S_B$ to every element of $C^i$.  
If $P$ is empty, then the fundamental circuit of $e$ also includes an element of $B^{i-1}$ and we are done.  
Otherwise,  as noted above, the fundamental circuit of $f$ in $B$ contains all of $P$ together with at least one element of $C^i$ and at least one element of $B^{i-1}$.  Thus $f$ is joined by an edge of $S_B$ to every element of $P$, and to at least one element of $C^i$ and to at least one element of $B^{i-1}$.
Therefore, in $S_B$ every element of $B^i$ is within distance 4 of some element of $B^{i-1}$.
\end{proof}

\section{Conclusion}
We studied the reconfiguration of labelled bases of a rank $r$ matroid and provided an upper bound of $O(r\log r)$ on the worst-case reconfiguration distance for graphic matroids, and a bound of $O(r^{1.5})$ for general matroids. The obvious next question is whether this is tight. The only lower bound we have so far is $\Omega(r)$.

%Plane spanning trees.
Another natural question is to find the 
minimum number of basis exchange steps needed to transform one given labelled basis to another.  It an open question whether this problem is NP-hard or polynomial-time solvable.

%Note that an element labelled the same in both bases might need to leave the basis during reconfiguration (unlike the unlabelled case).

\section*{Acknowledgements}

We could not have proved these results without major help from Jim Geelen.
The results first appeared in the second author's PhD thesis~\cite{Pathak}.  Theorem~\ref{theorem-spanning-tree-characterization} was proved jointly with Sander Vendonschot and Prosenjit Bose.

%%%%%%%%%%%%%%%%%%%%%%%%%%%%%%%%%%
\bibliographystyle{abuser}   
\bibliography{matroid-reconfig}

\begin{thebibliography}{10}

\bibitem{BH09}
P.~Bose and F.~Hurtado.
\newblock Flips in planar graphs.
\newblock {\em Computational Geometry: Theory and Applications} 42:60--80,
  2009, \href{http://dx.doi.org/10.1016/j.comgeo.2008.04.001}%
{doi:10.1016/j.comgeo.2008.04.001}.

\bibitem{BLPV13}
P.~Bose, A.~Lubiw, V.~Pathak, and S.~Verdonschot.
\newblock Flipping edge-labelled triangulations.
\newblock {\em CoRR}, 2013, \url{arXiv:1310.1166v2}.

\bibitem{HHMR03}
C.~Hernando, F.~Hurtado, M.~Mora, and E.~Rivera-Campo.
\newblock Grafos de {\'a}rboles etiquetados y grafos de {\'a}rboles
  geom{\'e}tricos etiquetados.
\newblock {\em Proc. X Encuentros de Geometra Computacional}, pp.~13-19, 2003.

\bibitem{van-den-Heuvel}
J.~van~den Heuvel.
\newblock The complexity of change.
\newblock {\em Surveys in Combinatorics}, vol. 409, pp.~127--160. Cambridge
  University Press, 2013.

\bibitem{Ito-2011}
T.~Ito, E.~D. Demaine, N.~J.~A. Harvey, C.~H. Papadimitriou, M.~Sideri,
  R.~Uehara, and Y.~Uno.
\newblock On the complexity of reconfiguration problems.
\newblock {\em Theoretical Computer Science} 412(12):1054--1065, 2011,
  \href{http://dx.doi.org/10.1016/j.tcs.2010.12.005}%
{doi:10.1016/j.tcs.2010.12.005}.

\bibitem{MRWW05}
N.~McMurray, T.~J. Reid, B.~Wei, and H.~Wu.
\newblock Largest circuits in matroids.
\newblock {\em Advances in Applied Mathematics} 34(1):213--216, 2005,
  \href{http://dx.doi.org/10.1016/j.aam.2004.09.002}%
{doi:10.1016/j.aam.2004.09.002}.

\bibitem{MR08}
W.~Mulzer and G.~Rote.
\newblock Minimum-weight triangulation is {NP}-hard.
\newblock {\em J. ACM} 55(2):11:1--11:29, 2008,
  \href{http://dx.doi.org/10.1145/1346330.1346336}%
{doi:10.1145/1346330.1346336}.

\bibitem{Oxley-2011}
J.~Oxley.
\newblock {\em Matroid Theory, second edition}.
\newblock Oxford University Press, 2011.

\bibitem{Pathak}
V.~Pathak.
\newblock {\em Reconfiguring Triangulations}.
\newblock Ph.D. thesis, University of Waterloo, 2014.

\bibitem{Schrijver}
A.~Schrijver.
\newblock {\em Combinatorial Optimization: Polyhedra and Efficiency}.
\newblock Springer, 2002.

\bibitem{Tutte-linking}
W.~T. Tutte.
\newblock Menger{'}s theorem for matroids.
\newblock {\em J. Res. Nat. Bur. Standards Sect. B} 69:49--53, 1965.

\end{thebibliography}

\end{document}